\documentclass[a4paper,UKenglish]{lipics-v2016}
 
\usepackage{microtype}
\usepackage{bbm}
\usepackage{graphicx}
\usepackage{ragged2e}
\usepackage{algorithm}
\usepackage{wrapfig}
\usepackage{algpseudocode}
\usepackage{placeins}
\usepackage[unq]{unique}
\usepackage{color}
\usepackage{mathalfa}

\let\oldnl\nl
\newcommand{\nonl}{\renewcommand{\nl}{\let\nl\oldnl}}

\title{Scheduling Distributed Clusters of Parallel Machines : Primal-Dual and LP-based Approximation Algorithms [Full Version] \footnote{All authors conducted this work at the University of Maryland, College Park. This work was made possible by the National Science Foundation, REU Grant CCF 1262805, and the Winkler Foundation. This work was also partially supported by NSF Grant CCF 1217890.}}
\titlerunning{Scheduling Distributed Clusters of Parallel Machines [Full Version]} 

\author[1]{Riley Murray}
\author[2]{Samir Khuller}
\author[3]{Megan Chao}
\affil[1]{Department of Industrial Engineering \& Operations Research, University of California, Berkeley \\
Berkeley, CA 94709, USA\\
 \texttt{rjmurray@berkeley.edu}}
\affil[2]{Department of Computer Science, University of Maryland, College Park \\
  College Park, MD 20742, USA\\
  \texttt{samir@cs.umd.edu}}
\affil[3]{Department of Electrical Engineering \& Computer Science, Massachusetts Institute of Technology \\
  50 Vassar St, Cambridge, MA 02142, USA\\
  \texttt{megchao@mit.edu}}
\authorrunning{R. Murray, S. Khuller, and M. Chao} 

\Copyright{Riley Murray, Samir Khuller, Megan Chao}


\EventEditors{Piotr Sankowski and Christos Zaroliagis}
\EventNoEds{2}
\EventLongTitle{24rd Annual European Symposium on Algorithms (ESA 2016)}
\EventShortTitle{ESA 2016}
\EventAcronym{ESA}
\EventYear{2016}
\EventDate{August 22--24, 2016}
\EventLocation{Aarhus, Denmark}
\EventLogo{}
\SeriesVolume{57}
\ArticleNo{234} 
\begin{document}

\maketitle

\begin{abstract}
The Map-Reduce computing framework rose to prominence with datasets of such size that dozens of machines on a single cluster were needed for individual jobs. 
As datasets approach the exabyte scale, a single job may need distributed processing not only on multiple machines, but on multiple \textit{clusters}.
We consider a scheduling problem to minimize weighted average completion time of $n$ jobs on $m$ distributed clusters of parallel machines. 
In keeping with the scale of the problems motivating this work, we assume that (1) each job is divided into $m$ ``subjobs'' and (2) distinct subjobs of a given job may be processed concurrently. 

When each cluster is a single machine, this is the NP-Hard \textit{concurrent open shop} problem.
A clear limitation of such a model is that a serial processing assumption sidesteps the issue of how different tasks of a given subjob might be processed in parallel. 
Our algorithms explicitly model clusters as pools of resources and effectively overcome this issue.

Under a variety of parameter settings, we develop two constant factor approximation algorithms for this problem. The first algorithm uses an LP relaxation tailored to this problem from prior work. This LP-based algorithm provides strong performance guarantees. Our second algorithm exploits a surprisingly simple mapping to the special case of one machine per cluster.
This mapping-based algorithm is combinatorial and extremely fast. These are the first constant factor approximations for this problem.

\textit{Remark - A shorter version of this paper (one that omitted several proofs) appeared in the proceedings of the 2016 European Symposium on Algorithms.}
\end{abstract}

\section{Introduction}\label{sec:intro}
%
%

It is becoming increasingly impractical to store full copies of large datasets on more than one data center \cite{Hajjat2012}. As a result, the data for a single job may be located not on multiple machines, but on multiple \textit{clusters} of machines. 
To maintain fast response-times and avoid excessive network traffic, it is advantageous to perform computation for such jobs in a completely distributed fashion \cite{HGY}.
In addition, commercial platforms such as AWS Lambda and Microsoft's Azure Service Fabric are demonstrating a trend of centralized cloud computing frameworks in which the user manages neither data flow nor server allocation \cite{AWSLambda, Azure}. 
In view of these converging issues, the following scheduling problem arises:

\textit{If computation is done locally to avoid excessive network traffic, how can individual clusters on the broader grid coordinate schedules for maximum throughput? }

This was precisely the motivation for Hung, Golubchik, and Yu in their 2015 ACM Symposium on Cloud Computing paper \cite{HGY}.  
Hung et al. modeled each cluster as having an arbitrary number of identical parallel machines, and choose an objective of average job completion time. 
As such a problem generalizes the NP-Hard concurrent open shop problem, they proposed a heuristic approach. 
Their heuristic (called ``SWAG'') runs in $O(n^2m)$ time and performed well on a variety of data sets. Unfortunately, SWAG offers poor worst-case performance, as we show in Section \ref{sec:TSPT}.

Our contributions to this problem are to extend the model considered by Hung et al. and to introduce the first constant-factor approximation algorithms for this general problem. 
Our extensions of Hung et al.'s model are (1) to allow different machines within the same cluster to operate at different speeds, (2) to incorporate pre-specified ``release times'' (times before which a subjob cannot be processed), and (3) to support \textit{weighted} average job completion time.
We present two algorithms for the resulting problem.
Our combinatorial algorithm exploits a surprisingly simple mapping to the special case of one machine per cluster, where the problem can be approximated in $O(n^2 + nm)$ time. We also present an LP-rounding approach with strong performance guarantees. E.g., a 2-approximation when machines are of unit speed and subjobs are divided into equally sized (but not necessary \textit{unit}) tasks.

\subsection{Formal Problem Statement}\label{subsec:ps}
\begin{definition}[Concurrent Cluster Scheduling]{\color{white} . } \hfill 
\vspace{2pt}
\begin{itemize}
\item There is a set $M$ of $m$ clusters, and a set $N$ of $n$ jobs. 
For each job $j \in N$, there is a set of $m$ ``subjobs'' (one for each cluster).

\item Cluster $i \in M$ has $m_i$ parallel machines, and machine $\ell$ in cluster $i$ has speed $v_{\ell i}$. 
Without loss of generality, assume $v_{\ell i}$ is decreasing in $\ell$.
\footnote{Where we write ``decreasing'', we mean ``non-increasing.'' Where we write ``increasing'', we mean ``non-decreasing''.} 

\item The $i^{\text{th}}$ subjob for job $j$ is specified by a set of tasks to be performed by machines in cluster $i$, denote this set of tasks $T_{ji}$. 
For each task $t \in T_{ji}$, we have an associated processing time $p_{jit}$ (again w.l.o.g., assume $p_{jit}$ is decreasing in $t$). 
We will frequently refer to ``the subjob of job $j$ at cluster $i$'' as ``subjob $(j,i)$.''

\item Different subjobs of the same job may be processed concurrently on different clusters. 
\item Different tasks of the same subjob may be processed concurrently on different machines within the same cluster.
\item A subjob is complete when all of its tasks are complete, and a job is complete when all of its subjobs are complete. We denote a job's completion time by ``$C_j$''.
\item The objective is to minimize weighted average job completion time (job $j$ has weight $w_j$).
\item For the purposes of computing approximation ratios, it is equivalent to minimize $\sum w_j C_j$. We work with this equivalent objective throughout this paper.
\end{itemize}
\end{definition}

A machine is said to operate at \textit{unit speed} it if can complete a task with processing requirement ``$p$'' in $p$ units of time. More generally, a machine with speed ``$v$'' ($v \geq 1$) processes the same task in $p/ v$ units of time. Machines are said to be \textit{identical} if they are all of unit speed, and \textit{uniform} if they differ only in speed.

In accordance with Graham et al.'s $\alpha|\beta|\gamma$ taxonomy for scheduling problems \cite{Graham1979} we take $\alpha = CC$ to refer to the concurrent cluster environment, and denote our problem by $CC|| \sum w_jC_j$.\footnote{ A problem $\alpha|\beta|\gamma$ implies a particular environment $\alpha$, objective function $\gamma$, and optional constraints $\beta$.} Optionally, we may associate a release time $r_{ji}$ to every subjob.  If any subjobs are released after time zero, we write $CC| r | \sum w_jC_j$.

\subsubsection{Example Problem Instances}\label{subsubsec:introToModel}

We now illustrate our model with several examples (see Figures \ref{fig:probstm1and2} and \ref{fig:probstm3and4}). The tables at left have rows labeled to identify jobs, and columns labeled to identify clusters; each entry in these tables specifies the processing requirements for the corresponding subjob. The diagrams to the right of these tables show how the given jobs might be scheduled on clusters with the indicated number of machines. 
\begin{figure}[ht!]
\centering
  \includegraphics[width=0.8\linewidth]{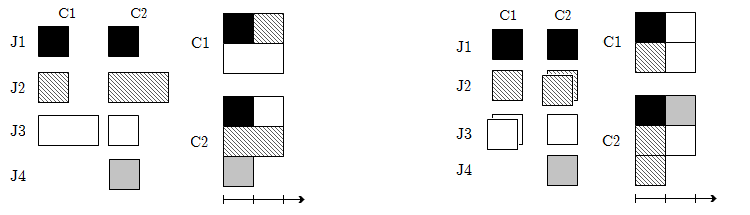}
	\caption{Two examples of our scheduling model. \textbf{Left}: Our baseline example. There are 4 jobs and 2 clusters. Cluster 1 has 2 identical machines, and cluster 2 has 3 identical machines. Note that job 4 has no subjob for cluster 1 (this is permitted within our framework). In this case every subjob has at most one task. \textbf{Right}: Our baseline example with a more general subjob framework : subjob (2,2) and subjob (3,1) both have two tasks. The tasks shown are unit length, but our framework \textit{does not} require that subjobs be divided into equally sized tasks. }
  \label{fig:probstm1and2}
\end{figure}
\begin{figure}[ht!]
  \includegraphics[width=\linewidth]{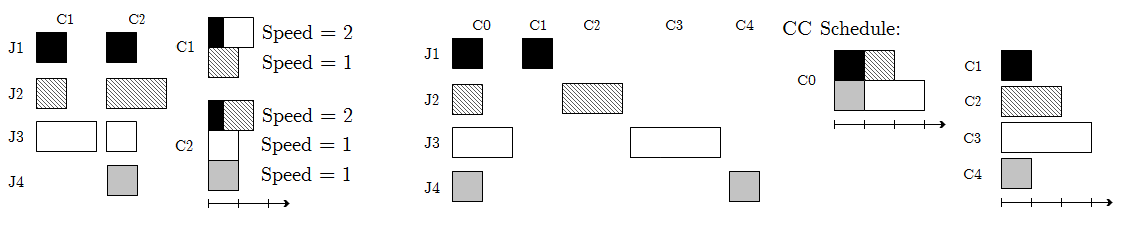}
	\caption{Two additional examples of our model. \textbf{Left}: Our baseline example, with variable machine speeds. Note that the benefit of high machine speeds is only realized for tasks assigned to those machines in the final schedule. \textbf{Right}: A problem with the peculiar structure that (1) all clusters but one have a single machine, and (2) most clusters have non-zero processing requirements for only a single job. We will use such a device for the total weighted lateness reduction in Section \ref{sec:relationshipsBetweenProbs}.}
  \label{fig:probstm3and4}
\end{figure}

\subsection{Related Work}

Concurrent cluster scheduling subsumes many fundamental machine scheduling problems. For example, if we restrict ourselves to a single cluster (i.e. $m = 1$) we can schedule a set of jobs on a bank of identical parallel machines to minimize makespan ($C_{\max}$) or total weighted completion time ($\sum w_j C_j$). With a more clever reduction, we can even minimize \textit{total weighted lateness} ($\sum w_j L_j$) on a bank of identical parallel machines (see Section \ref{sec:relationshipsBetweenProbs}). Alternatively, with $m > 1$ but $\forall i \in M, m_i = 1$, our problem reduces to the well-studied ``concurrent open shop'' problem.

Using Graham et al.'s taxonomy, the concurrent open shop problem is written as $PD||\sum w_j C_j$. Three groups \cite{Chen2000, Garg2007, llp} independently discovered an LP-based 2-approximation for $PD||\sum w_j C_j$ using the work of Queyranne \cite{Queyranne1993}. The linear program in question has an exponential number of constraints, but can still be solved in polynomial time with a variant of the Ellipsoid method. Our ``strong'' algorithm for concurrent cluster scheduling refines the techniques contained therein, as well as those of Schulz \cite{Schulz1996, Schulz2012} (see Section \ref{sec:lpAlg}).

Mastrolilli et al. \cite{mqssu} developed a primal-dual algorithm for $PD || \sum w_j C_j$ that does not use LP solvers. ``MUSSQ''\footnote{A permutation of the author's names: Mastrolilli, Queyranne, Schulz, Svensson, and Uhan.} is significant for both its speed and the strength of its performance guarantee : it achieves an approximation ratio of 2 in only $O(n^2 + nm)$ time. Although MUSSQ does not require an LP solver, its proof of correctness is based on the fact that it finds a feasible solution to the dual a particular linear program. Our ``fast'' algorithm for concurrent cluster scheduling uses MUSSQ as a subroutine (see Section \ref{sec:TSPT}).

Hung, Golubchik, and Yu \cite{HGY} presented a framework designed to improve scheduling across geographically distributed data centers. The scheduling framework had a centralized scheduler (which determined a job ordering) and local dispatchers which carried out a schedule consistent with the controllers job ordering. Hung et al. proposed a particular algorithm for the controller called ``SWAG.'' SWAG performed well in a wide variety of simulations where each data center was assumed to have the same number of identical parallel machines. We adopt a similar framework to Hung et al., but we show in Section \ref{subsec:swagDegenerate} that SWAG has no constant-factor performance guarantee.

\subsection{Paper Outline \& Algorithmic Results}\label{subsec:outlineAndResults}

Although only one of our algorithms requires \textit{solving} a linear program, both algorithms use the same linear program in their proofs of correctness; we introduce this linear program in Section \ref{sec:introduceLP} before discussing either algorithm. Section \ref{sec:listSched} establishes how an ordering of jobs can be processed to completely specify a schedule. This is important because the complex work in both of our algorithms is to generate an ordering of jobs for each cluster.

Section \ref{sec:lpAlg} introduces our ``strong'' algorithm: CC-LP. CC-LP can be applied to any instance of concurrent cluster scheduling, including those with non-zero release times $r_{ji}$. A key in CC-LP's strong performance guarantees lay in the fact that it allows different permutations of subjobs for different clusters. By providing additional structure to the problem (but while maintaining a generalization of concurrent open shop) CC-LP becomes a 2-approximation. This is significant because it is NP-Hard to approximate concurrent open shop (and by extension, our problem) with ratio $2-\epsilon$ for any $\epsilon > 0$ \cite{nphard2}.

Our combinatorial algorithm (``CC-TSPT'') is presented in Section \ref{sec:TSPT}. The algorithm is fast, provably accurate, and has the interesting property that it can schedule all clusters using the same permutation of jobs.\footnote{We call such schedules ``single-$\sigma$ schedules.'' As we will see later on, CC-TSPT serves as a constructive proof of existence of near-optimal single-$\sigma$ schedules for all instances of $CC||\sum w_j C_j$, \textit{including} those instances for which single-$\sigma$ schedules are strictly sub-optimal. This is addressed in Section \ref{sec:discAndConc}.} After considering CC-TSPT in the general case, we show how fine-grained approximation ratios can be obtained in the ``fully parallelizable'' setting of Zhang et al. \cite{zwl}. We conclude with an extension of CC-TSPT that maintains performance guarantees while offering improved empirical performance.

The following table summarizes our results for approximation ratios. For compactness, condition $Id$ refers to identical machines (i.e. $v_{\ell i}$ constant over $\ell$), condition $A$ refers to $r_{ji} \equiv 0$, and condition $B$ refers to $p_{jit} \text{ constant over } t \in T_{ji}$.
\begin{center}
\begin{tabular}{l| cccccc}
\hline
 		& $(Id,A,B)$ & $(Id, \neg A, B)$ & $(Id,A,\neg B)$ & $(Id,\neg A, \neg B)$ & $(\neg Id, A)$ & $(\neg Id, \neg A)$ \\ \hline
CC-LP	&	2	  &     3         &    3         &   4  &  $2+R$     &  $3+R$    \\
CC-TSPT &   3     &     -          &   3         &   -  &  $2+R$      &    -  \\ \hline
\end{tabular}
\end{center}
The term $R$ is the maximum over $i$ of $R_i$, where $R_i$ is the ratio of fastest machine to \textit{average} machine speed at cluster $i$.

The most surprising of all of these results is that our scheduling algorithms are remarkably simple. The first algorithm solves an LP,
and then the scheduling can be done easily on each cluster. The second algorithm is again a rather surprising simple reduction to the
case of one machine per cluster (the well understood concurrent open shop problem) and yields a simple combinatorial algorithm. The proof
of the approximation guarantee is somewhat involved however.

In addition to algorithmic results, we demonstrate how our problem subsumes that of minimizing total weighted lateness on a bank of identical parallel machines (see Section \ref{sec:relationshipsBetweenProbs}). Section \ref{sec:discAndConc} provides additional discussion and highlights our more novel technical contributions.

\section{The Core Linear Program }\label{sec:introduceLP}



Our linear program has an unusual form. Rather than introduce it immediately, we conduct a brief review of prior work on similar LP's. All the LP's we discuss in this paper have objective function $\sum w_j C_j$, where $C_j$ is a decision variable corresponding to the completion time of job $j$, and $w_j$ is a weight associated with job $j$. 

\textit{For the following discussion only, we adopt the notation in which job $j$ has processing time $p_j$. In addition, if multiple machine problems are discussed, we will say that there are $\mathsf{m}$ such machines (possibly with speeds $s_i, i \in \{1,\ldots, \mathsf{m}\}$).} 

The earliest appearance of a similar linear program comes from Queyranne \cite{Queyranne1993}. In his paper, Queyranne presents an LP relaxation for sequencing $n$ jobs on a single machine where all constraints are of the form $\sum_{j \in S} p_j C_j \geq \frac{1}{2}\left[\left(\sum_{j \in S} p_j \right)^2 + \sum_{j \in S} p_j^2\right]$ where $S$ is an arbitrary subset of jobs. Once a set of optimal $\{C_j^\star\}$ is found, the jobs are scheduled in increasing order of $\{C_j^\star\}$. These results were primarily theoretical, as it was known at his time of writing that sequencing $n$ jobs on a single machine to minimize $\sum w_j C_j$ can be done optimally in $O(n \log n)$ time.

Queyranne's constraint set became particularly useful for problems with \textit{coupling} across distinct machines (as occurs in concurrent open shop). Four separate groups \cite{Chen2000,Garg2007,llp, mqssu} saw this and used the following LP in a 2-approximation for concurrent open shop scheduling.
\begin{equation}
(\text{LP0}) ~~ \min \sum_{j \in N} w_j C_j ~~ \text{s.t.} ~~ \textstyle\sum_{j \in S} p_{ji} C_j \geq 
		\frac{1}{2}
		\left[ 
			\left(\textstyle\sum_{j \in S} p_{ji}\right)^2 + \left(\textstyle\sum_{j \in S} p_{ji}^2\right) 
		\right] ~ \forall ~ \substack{S \subseteq N \\  i \in M }\nonumber
\end{equation}
In view of its tremendous popularity, we sometimes refer to the linear program above as the \textit{canonical relaxation} for concurrent open shop.

Andreas Schulz's Ph.D. thesis developed Queyranne's constraint set in greater depth \cite{Schulz1996}. As part of his thesis, Schulz considered scheduling $n$ jobs on $\mathsf{m}$ identical parallel machines with constraints of the form $\sum_{j \in S} p_j C_j \geq \frac{1}{2\mathsf{m}} \left(\sum_{j \in S} p_j \right)^2 + \frac{1}{2}\sum_{j \in S} p_j^2$. In addition, Schulz showed that the constraints $\sum_{j \in S} p_j C_j \geq \left[2 \sum_{i=1}^{\mathsf{m}} s_i \right]^{-1}\left[\left(\sum_{j \in S} p_j \right)^2 + \sum_{j \in S} p_j^2\right]$ are satisfied by any schedule of $n$ jobs on $\mathsf{m}$ uniform machines. In 2012, Schulz refined the analysis for several of these problems \cite{Schulz2012}. For constructing a schedule from the optimal $\{C_j^\star\}$, Schulz considered scheduling jobs by increasing order of $\{C_j^\star\}$, $\{C_j^\star - p_j/2\}$, and $\{C_j^\star - p_j/(2\mathsf{m})\}$.

\subsection{Statement of LP1}

The model we consider allows for more fine-grained control of the job structure than is indicated by the LP relaxations above. Inevitably, this comes at some expense of simplicity in LP formulations. In an effort to simplify notation, we define the following constants, and give verbal interpretations for each. 
\begin{equation}
 \mu_{i} \doteq \textstyle\sum_{\ell = 1}^{m_i} v_{\ell i} \qquad q_{ji} \doteq \min{\lbrace|T_{ji}|, m_i\rbrace}  \qquad \mu_{ji} \doteq \textstyle\sum_{\ell = 1}^{q_{ji}} v_{\ell i} \qquad p_{ji} \doteq \textstyle\sum_{t \in T_{ji}} p_{jit}
\end{equation}
From these definitions, $\mu_i$ is the processing power of cluster $i$. For subjob $(j,i)$, $q_{ji}$ is the maximum number of machines that could process the subjob, and $\mu_{ji}$ is the maximum processing power than can be brought to bear on the same. Lastly, $p_{ji}$ is the total processing requirement of subjob $(j,i)$. In these terms, the core linear program, LP1, is as follows.
\begin{align*}
	\text{(LP1) } \min &\textstyle\sum_{j \in N} w_j C_j \\
	s.t.\quad (1A) \quad & \textstyle\sum_{j \in S} p_{ji} C_j 
	        \geq \frac{1}{2} 
	        \left[ 
	           \left(\textstyle\sum_{j \in S} p_{ji}\right)^2/\mu_i
                + \textstyle\sum_{j \in S} p_{ji}^2/\mu_{ji} 
            \right]  \qquad ~\forall S \subseteq N, i \in M \\
	(1B) \quad	& C_j \geq p_{jit}/v_{1i} + r_{ji} \qquad ~\forall  i \in M,~ j \in N,~ t \in T_{ji}\\
	(1C) \quad	& C_j \geq p_{ji}/\mu_{ji} + r_{ji} \qquad ~\forall j \in N,~ i \in M 	 
\end{align*} 

Constraints ($1A$) are more carefully formulated versions of the polyhedral constraints introduced by Queyranne \cite{Queyranne1993} and developed by Schulz \cite{Schulz1996}. The use of $\mu_{ji}$ term is new and allows us to provide stronger performance guarantees for our framework where subjobs are composed of \textit{sets} of tasks. As we will see, this term is one of the primary factors that allows us to parametrize results under varying machine speeds in terms of maximum to \textit{average} machine speed, rather than maximum to \textit{minimum} machine speed. Constraints ($1B$) and ($1C$) are simple lower bounds on job completion time.

The majority of this section is dedicated to proving that LP1 is a valid relaxation of $CC|r|\sum w_jC_j$. Once this is established, we prove the that LP1 can be solved in polynomial time by providing a separation oracle with use in the Ellipsoid method. Both of these proofs use techniques established in Schulz's Ph.D. thesis \cite{Schulz1996}. 

\subsection{Proof of LP1's Validity}




The lemmas below establish the basis for both of our algorithms. Lemma \ref{lem:sumOfSquaresDiffSpeeds} generalizes an inequality used by Schulz \cite{Schulz1996}. Lemma \ref{lem:feasForLP1} relies on Lemma \ref{lem:sumOfSquaresDiffSpeeds} and cites an inequality mentioned in the preceding section (and proven by Queyranne \cite{Queyranne1993}). 
\begin{lemma}
Let $\{a_1,\ldots a_z\}$ be a set of non-negative real numbers. We assume that $k \leq z$ of them are positive. Let $b_i$ be a set of decreasing positive real numbers. Then
\begin{center}
$  \sum_{i = 1}^z a_i^2 / b_i  \geq \left(\sum_{i = 1}^z a_i \right)^2 / \left(\sum_{i = 1}^k b_i\right)$.
\end{center}
\label{lem:sumOfSquaresDiffSpeeds}
\end{lemma}
\begin{proof}\footnote{The proceedings version of this paper stated that the proof cites the AM-GM inequality and proceeds by induction from $z=k=2$. We have opted here to demonstrate a different (simpler) proof that we discovered only after the proceedings version was finalized.}
We only show the case where $k=z$. Define $\mathbf{a} = [a_1,\ldots,a_k] \in \mathbb{R}^k_+$, $\mathbf{b} = [b_1, \ldots, b_k ] \in \mathbb{R}^k_{++}$, and $\mathbbm{1}$ as the vector of $k$ ones. Now, set $\mathbf{u} = \mathbf{a} / \sqrt{\mathbf{b}}$ and $\mathbf{w} = \sqrt{\mathbf{b}}$ (element-wise), and note that $\langle \mathbf{a}, \mathbbm{1} \rangle = \langle \mathbf{u}, \mathbf{w} \rangle$. In these terms, it is clear that $(\sum_{i=1}^k a_i)^2 = \langle \mathbf{u}, \mathbf{w} \rangle^2$.

Given this, one need only cite Cauchy-Schwarz (namely, $\langle \mathbf{u}, \mathbf{w} \rangle^2 \leq \langle \mathbf{u}, \mathbf{u} \rangle \cdot \langle \mathbf{w}, \mathbf{w} \rangle$) and plug in the definitions of $\mathbf{u}$ and $\mathbf{w}$ to see the desired result.
\end{proof}

\begin{lemma}[Validity Lemma]
Every feasible schedule for an instance $I$ of $CC|r|\sum w_jC_j$ has completion times that define a feasible solution to LP1($I$). \label{lem:feasForLP1}
\end{lemma}
\begin{proof}
As constraints ($1B$) and ($1C$) are clear lower bounds on job completion time, it suffices to show the validity of constraint ($1A$). Thus, let $S$ be a non-empty subset of $N$, and fix an arbitrary but feasible schedule ``$F$'' for $I$. 

Define $C^{F}_{ji}$ as the completion time of subjob $(j,i)$ under schedule $F$. Similarly, define $C^{F}_{ji\ell}$ as the first time at which tasks of subjob $(j,i)$ scheduled on machine $\ell$ of cluster $i$ are finished. Lastly, define $p^{\ell}_{ji}$ as the total processing requirement of job $j$ scheduled on machine $\ell$ of cluster $i$. Note that by construction, we have $C^{F}_{ji} = \max_{\ell \in \{1,\ldots,m_i\}}{C^{F}_{ji\ell}}$ and $C^F_j = \max_{i \in M}{C^F_{ji}}$. 
Since $p_{ji} = \sum_{\ell = 1}^{m_i} p^{\ell}_{ji}$, we can rather innocuously write
\begin{equation}
\textstyle\sum_{j \in S} p_{ji} C^{F}_{ji} = \textstyle\sum_{j \in S}\left[ \textstyle\sum_{\ell  = 1}^{m_i}  p^{\ell}_{ji}  \right] C^{F}_{ji} . 
\end{equation} 
But using $C^{F}_{ji} \geq C^{F}_{ji\ell}$, we can lower-bound $\sum_{j \in S} p_{ji} C^{F}_{ji}$. Namely,
\begin{equation}
\textstyle\sum_{j \in S} p_{ji} C^{F}_{ji} \geq \textstyle\sum_{j \in S}\textstyle\sum_{\ell = 1}^{m_i}  p^{\ell}_{ji} C^{F}_{ji\ell} =  \textstyle\sum_{\ell = 1}^{m_i} v_{\ell i}\textstyle\sum_{j \in S} \left[p^{\ell}_{ji}/v_{\ell i} \right]C^{F}_{ji\ell} \label{eq:specificMachines}
\end{equation}
The next inequality uses a bound on $\textstyle\sum_{j \in S}\left[p^{\ell}_{ji}/v_{\ell i} \right]C^{F}_{ji\ell}$ proven by Queyranne \cite{Queyranne1993} for any subset $S$ of $N$ jobs with processing times $\left[p^{\ell}_{ji}/v_{\ell i} \right]$ to be scheduled on a single machine.\footnote{Here, our machine is machine $\ell$ on cluster $i$.}
\begin{equation}
\textstyle\sum_{j \in S} \left[p^{\ell}_{ji}/v_{\ell i} \right]C^{F}_{ji\ell} \geq \frac{1}{2} \left[\left(\textstyle\sum_{j \in S} \left[p^{\ell}_{ji}/v_{\ell i} \right]\right)^2 + \textstyle\sum_{j \in S} \left(
\left[p^{\ell}_{ji}/v_{\ell i} \right]\right)^2 \right]\label{eq:queyranne}
\end{equation}
Combining inequalities \eqref{eq:specificMachines} and \eqref{eq:queyranne}, we have the following.
\begin{align}
\textstyle\sum_{j \in S} p_{ji} C^{F}_{ji} &\geq \frac{1}{2} \textstyle\sum_{\ell=1}^{m_i} v_{\ell i} \left[\left(\textstyle\sum_{j \in S} \left[p^{\ell}_{ji}/v_{\ell i} \right]\right)^2 + \textstyle\sum_{j \in S} \left(
\left[p^{\ell}_{ji}/v_{\ell i} \right]\right)^2 \right]  \\
& \geq \frac{1}{2} \left[\textstyle\sum_{\ell = 1}^{m_i}\left(\sum_{j \in S} p^\ell_{j i}\right)^2 / v_{\ell i} + \sum_{j \in S} \sum_{\ell = 1}^{m_i} \left(p^\ell_{j i}\right)^2 / v_{\ell i} \right] \label{eq:differentBoundForLP1}
\end{align}
Next, we apply Lemma \ref{lem:sumOfSquaresDiffSpeeds} to the right hand side of inequality \eqref{eq:differentBoundForLP1} a total of $|S|+1$ times.
\begin{align}
&\textstyle\sum_{\ell = 1}^{m_i} \left(\textstyle\sum_{j \in S} p^\ell_{j i}\right)^2 /v_{\ell i} \geq \left(\textstyle\sum_{\ell = 1}^{m_i}\textstyle\sum_{j \in S} p^\ell_{j i}\right)^2/ \textstyle\sum_{\ell = 1}^{m_i} v_{\ell i} = \left(\textstyle\sum_{j \in S} p_{ji}\right)^2 / \mu_i \\
&\textstyle\sum_{\ell = 1}^{m_i} \left(p^\ell_{j i}\right)^2 / v_{\ell i} \geq \left(\sum_{\ell = 1}^{m_i}p^\ell_{j i}\right)^2 / \textstyle\sum_{\ell = 1}^{q_{j i}} v_{\ell i} = p_{j i}^2/\mu_{j i} ~~\forall~ j \in S
\end{align}
Citing $C^{F}_{j} \geq C^{F}_{ji}$, we arrive at the desired result.
\begin{equation}
\textstyle\sum_{j \in S} p_{ji} C^{F}_{j} \geq \frac{1}{2}\left[\left(\textstyle\sum_{j \in S} p_{ji} \right)^2/\mu_i + \textstyle\sum_{j \in S}p_{ji}^2/\mu_{ji}\right] \qquad  \text{``constraint }(1A)\text{''}
\end{equation}

\end{proof}

\subsection{Theoretical Complexity of LP1}\label{subsec:thankGodPolyTime}
%
%
As the first of our two algorithms requires solving LP1 directly, we need to address the fact that LP1 has $m \cdot (2^n - 1) + n$ constraints. 
Luckily, it is still possible to such solve linear programs in polynomial time with the Ellipsoid method; we introduce the following separation oracle for this purpose.

\begin{definition}[Oracle LP1]
Define the \textit{violation}
\begin{equation}
V(S,i) = \frac{1}{2} \left[ 
	           \left(\textstyle\sum_{j \in S} p_{j i}\right)^2/\mu_i 
                + \textstyle\sum_{j \in S} p_{j i}^2/\mu_{j i} 
            \right] - \textstyle\sum_{j \in S} p_{j  i} C_{j}
\end{equation}
Let $\{C_j\} \in \mathbb{R}^n$ be a \textit{potentially} feasible solution to LP1. Let $\sigma_i$ denote the ordering when jobs are sorted in increasing order of $C_j - p_{j i}/(2\mu_{ji})$. Find the most violated constraint in $(1A)$ for $i \in M$ by searching over $V(S_i,i)$ for $S_i$ of the form $\{\sigma_i(1),\ldots,\sigma_i(j-1),\sigma_i(j)\},~ j \in \{1,\ldots,n\}$. If any of maximal $V(S_i^*,i) > 0$, then return $(S_i^*,i)$ as a violated constraint for ($1A$). Otherwise, check the remaining $n$ constraints $((1B)$ and $(1C))$ directly in linear time.
\end{definition} 

For fixed $i$, Oracle-LP1 finds the subset of jobs that maximizes ``violation'' for cluster $i$. That is, Oracle-LP1 finds $S_i^*$ such that $V(S_i^*,i) = \text{max}_{S \subset N} V(S,i)$. We prove the correctness of Oracle-LP1 by establishing a necessary and sufficient condition for a job $j$ to be in $S_i^*$.

\begin{lemma}
For $\mathbb{P}_i(A) \doteq \textstyle\sum_{j \in A} p_{ji}$, we have $x \in S_i^* \Leftrightarrow$ $ C_x - p_{xi}/(2\mu_{xi}) \leq \mathbb{P}_i(S_i^*)/\mu_i$. 
\label{lem:separation}
\end{lemma}
\begin{proof}
For given $S$ (not necessarily equal to $S_i^*$), it is useful to express $V(S,i)$ in terms of $V(S\cup x, i)$ or $V(S\setminus x, i)$ (depending on whether $x \in S$ or $x \in N \setminus S$). Without loss of generality, we restrict our search to $S : x \in S \Rightarrow p_{x,i} > 0$.

Suppose $ x \in S$. By writing $\mathbb{P}_i(S) = \mathbb{P}_i(S\setminus x) + \mathbb{P}_i(x)$, and similarly decomposing the sum $\textstyle\sum_{j \in S} p_{j i}^2/(2\mu_{ji})$, one can show the following.
\begin{align}
V(S, i) = & V(S\setminus x, i) + p_{xi}\left(\frac{1}{2}\left(\frac{2\mathbb{P}_i(S) - p_{xi}}{\mu_i} + \frac{p_{xi}}{\mu_{xi}}\right) - C_x \right) \label{eq:xInS}
\end{align}
Now suppose $ x \in N\setminus S$. In the same strategy as above (this time writing $ \mathbb{P}_i(S) = \mathbb{P}_i(S\cup x) - \mathbb{P}_i(x)$), one can show that
\begin{align}
V(S, i) =&  V(S\cup x, i) + p_{xi}\left(C_x - \frac{1}{2}\left(\frac{2\mathbb{P}_i(S) + p_{xi}}{\mu_i} + \frac{p_{xi}}{\mu_{xi}}\right) \right). \label{eq:xNotInS}
\end{align}
Note that Equations \eqref{eq:xInS} and \eqref{eq:xNotInS} hold for all $S$, including $S = S_i^*$. Turning our attention to $S_i^*$, we see that $x \in S_i^*$ implies that the second term in Equation \eqref{eq:xInS} is non-negative,  i.e. 
\begin{equation}
C_x - p_{xi}/(2\mu_{xi}) \leq \left(2\mathbb{P}_i(S_i^*) - p_{xi}\right)/(2\mu_i) < \mathbb{P}_i(S_i^*)/\mu_i.
\end{equation}
Similarly, $x \in N \setminus S_i^*$ implies the second term in Equation \eqref{eq:xNotInS} is non-negative.
\begin{equation}
C_x - p_{x i}/(2\mu_{x i}) \geq \left(2\mathbb{P}_i(S_i^*) + p_{x i}\right)/(2\mu_i) \geq \mathbb{P}_i(S_i^*)/\mu_i
\end{equation}
It follows that $x \in S_i^*$ iff $C_x - p_{x i}/(2\mu_{x i}) < \mathbb{P}_i(S_i^*)/\mu_i$.
\end{proof}

Given Lemma \ref{lem:separation}, It is easy to verify that sorting jobs in increasing order of $C_x - p_{xi}/(2\mu_{xi})$ to define a permutation $\sigma_i$ guarantees that $S_i^*$ is of the form $\{\sigma_i(1),\ldots,\sigma_i(j-1),\sigma_i(j)\}$ for some $j \in N$. This implies that for fixed $i$, Oracle-LP1 finds $S_i^*$ in $O(n \log(n))$ time. This procedure is executed once for each cluster, leaving the remaining $n$ constraints in $(1B)$ and $(1C)$ to be verified in linear time. Thus Oracle-LP1 runs in $O(mn\log(n))$ time.

By the equivalence of separation and optimization, we have proven the following theorem:
\begin{theorem}
LP1($I$) is a valid relaxation of $I \in \Omega_{CC}$, and is solvable in polynomial time. \label{thm:LP1feasAndSolve}
\end{theorem}

As was explained in the beginning of this section, linear programs such as those in \cite{Chen2000, Garg2007, llp, Queyranne1993, Schulz1996, Schulz2012} are processed with an appropriate sorting of the optimal decision variables $\{C^\star_j\}$. It is important then to have bounds on job completion times for a particular ordering of jobs. We address this next in Section \ref{sec:listSched}, and reserve our first algorithm for Section \ref{sec:lpAlg}.
\section{List Scheduling from Permutations}\label{sec:listSched}




The complex work in both of our proposed algorithms is to generate a \textit{permutation} of jobs. The procedure below takes such a permutation and uses it to determine start times, end times, and machine assignments for every task of every subjob.
\vspace{1em}

\noindent \textbf{List-LPT} : Given a single cluster with $m_i$ machines and a permutation of jobs $\sigma$, introduce $\text{List}(a,i) \doteq (p_{ai1}, p_{ai2},\ldots,p_{ai|T_{ai}|})$ as an ordered set of tasks belonging to subjob $(a,i)$, ordered by longest processing time first. Now define $\text{List}(\sigma) \doteq \text{List}(\sigma(1),i) \oplus \text{List}(\sigma(2),i) \oplus \cdots \oplus \text{List}(\sigma(n),i)$, where $\oplus$ is the concatenation operator. 

Place the tasks of $\text{List}(\sigma)$ in order- from the largest task of subjob $(\sigma(1),i)$, to the smallest task of subjob $(\sigma(n),i)$. When placing a particular task, assign it whichever machine and start time results in the task being completed as early as possible (without moving any tasks which have already been placed). Insert idle time (on all $m_i$ machines) as necessary if this procedure would otherwise start a job before its release time.
\vspace{1em}

The following Lemma is essential to bound the completion time of a set of jobs processed by List-LPT. The proof is adapted from Gonzalez et al. \cite{Gonzalez1977}.
\begin{lemma}
Suppose $n$ jobs are scheduled on cluster $i$ according to List-LPT($\sigma$). Then for $ \bar{v_i} \doteq \mu_i/m_i$, the completion time of subjob $(\sigma(j),i)$ $($denoted $C_{\sigma(j)i}$ $)$ satisfies
\begin{align}
&C_{\sigma(j)i} \leq \max_{1\leq k \leq j}{r_{\sigma(k)i}} + p_{\sigma(j)i1}/\bar{v_i} +  \left(\textstyle\sum_{k=1}^{j} p_{\sigma(k)i} - p_{\sigma(j)i1}\right)/\mu_i \label{eq:generalGonzalezLemma}
\end{align} \label{lem:CompTimesOnUniformMachines}
\end{lemma}
\begin{proof}
For now, assume all jobs are released at time zero. Let the task of subjob $(\sigma(j),i)$ to finish last be denoted $t^*$. If $t^*$ is not the task in $T_{\sigma(j)i}$ with least processing time, then construct a new set $T'_{\sigma(j)i} = \{ t :  p_{\sigma(j)it^*} \leq p_{\sigma(j)it} \} \subset T_{\sigma(j)i}$. Because the tasks of subjob $(\sigma(j),i)$ were scheduled by List-LPT (i.e. longest-processing-time-first), the sets of potential start times and machines for task $t^*$ (and hence the set of potential completion times for task $t^*$) are the same regardless of whether subjob $(\sigma(j),i)$ consisted of tasks $T_{\sigma(j)i}$ or the subset $T'_{\sigma(j)i}$. Accordingly, reassign $T_{\sigma(j)i} \leftarrow T'_{\sigma(j)i}$ without loss of generality.

Let $D_{\ell}^j$ denote the total demand for machine $\ell$ (on cluster $i$) once all tasks of subjobs $(\sigma(1),i)$ through $(\sigma(j-1),i)$ and all tasks in the set $T_{\sigma(j)i}\setminus \{t^*\}$ are scheduled. Using the fact that $C_{\sigma(j)i}v_{\ell i} \leq ({D}_{\ell}^{j} + p_{\sigma(j)i t^*}) \forall \ell \in \{1,\ldots,m_i\}$, sum the left and right and sides over $\ell$. This implies $C_{\sigma(j)i}\left( \textstyle\sum_{\ell = 1}^{m_i} v_{\ell i} \right) \leq ~ m_i p_{\sigma(j) i t^*} + \textstyle\sum_{\ell = 1}^{m_i} {D}_{\ell}^{j}$. Dividing by the sum of machine speeds and using the definition of $\mu_i$ yields
\begin{equation}
C_{\sigma(j)i} 
	~ \leq ~  m_i p_{\sigma(j)i t^*}/\mu_i + \textstyle\sum_{\ell = 1}^{m_i} {D}_{\ell}^j /\mu_i ~ \leq ~ p_{\sigma(j)i 1}/\bar{v_i} + \left(\textstyle\sum_{k = 1}^{j} p_{\sigma(k)i} - p_{\sigma(j)i1}\right)/\mu_i \label{eq:mainGonzalezLemma}
\end{equation}
where we estimated $p_{\sigma(j)i t^*}$ upward by $p_{\sigma(j)i 1}$. Inequality \eqref{eq:mainGonzalezLemma} completes our proof in the case when $r_{ji} \equiv 0$. 

Now suppose that some $r_{ji} > 0$. We take our policy to the extreme and suppose that all machines are left idle until every one of jobs $\sigma(1)$ through $\sigma(j)$ are released; note that this occurs precisely at time $\max_{1 \leq k \leq j} r_{\sigma(k)i}$. It is clear that beyond this point in time, we are effectively in the case where all jobs are released at time zero, hence we can bound the remaining time to completion by the right hand side of Inequality \ref{eq:mainGonzalezLemma}. As Inequality \ref{eq:generalGonzalezLemma} simply adds these two terms, the result follows.
\end{proof}

Lemma \ref{lem:CompTimesOnUniformMachines} is cited directly in the proof of Theorem \ref{thm:uniformLP1} and Lemma \ref{lem:TSPTboundInTermsOfPDandOPT}. Lemma \ref{lem:CompTimesOnUniformMachines} is used implicitly in the proofs of Theorems \ref{thm:identLP}, \ref{thm:identLP_2appxWithConstantTasks}, and \ref{thm:tspt_unit_tasks}.
\section{An LP-based Algorithm}\label{sec:lpAlg}

%
%

In this section we show how LP1 can be used to construct near optimal schedules for concurrent cluster scheduling both when $r_{ji} \equiv 0$ and when some $r_{ji} > 0$. Although solving LP1 is somewhat involved, the algorithm itself is quite simple:
\vspace{0.5em}

\noindent \textbf{Algorithm CC-LP} : Let $I = (T, r, w, v)$ denote an instance of $CC | r | \sum w_j C_j$. Use the optimal solution $\{C_j^\star\}$ of LP1($I$) to define $m$ permutations $\{\sigma_i : i \in M\}$ which sort jobs in increasing order of $C^\star_j - p_{ji}/(2\mu_{ji})$. For each cluster $i$, execute List-LPT($\sigma_i$).
\vspace{0.5em}

Each theorem in this section can be characterized by how various assumptions help us cancel an additive term\footnote{``$+p_{xit^*}$''; see associated proofs.} in an upper bound for the completion time of an arbitrary subjob $(x,i)$. Theorem \ref{thm:uniformLP1} is the most general, while Theorem \ref{thm:identLP_2appxWithConstantTasks} is perhaps the most surprising.

\subsection{CC-LP for Uniform Machines}\label{sec:unifLP}
\begin{theorem}
Let $\hat{C}_j$ be the completion time of job $j$ using algorithm CC-LP, and let $R$ be as in Section \ref{subsec:outlineAndResults}. If $r_{ji} \equiv 0$, then $ \textstyle\sum_{j \in N} w_j \hat{C}_j \leq \left(2 + R\right)OPT $. Otherwise, $ \textstyle\sum_{j \in N} w_j \hat{C}_j \leq \left(3 + R\right)OPT$.
\label{thm:uniformLP1}
\end{theorem}
\begin{proof}
For $y \in \mathbb{R}$, define $y^+ = \max\{y,0\}$. Now let $x \in N$ be arbitrary, and let $i \in M$ be such that $p_{xi} > 0$ (but otherwise arbitrary). Define $t^*$ as the last task of job $x$ to complete on cluster $i$, and let $j_i$ be such that $\sigma_i(j_i) = x$. Lastly, denote the optimal LP solution $\{C_j\}$.\footnote{We omit the customary $\star$ to avoid clutter in notation.} Because $\{C_j\}$ is a feasible solution to LP1, constraint $(1A)$ implies the following (set $S_i = \{\sigma_i(1),\ldots,\sigma_i(j_i - 1),x\}$)
\begin{align}
\frac{\left( \textstyle\sum_{k = 1}^{j_i} p_{\sigma_i(k)i} \right)^2}{2\mu_i}
	&\leq \sum_{k = 1}^{j_i} p_{\sigma_i(k)i}\left(C_{\sigma_i(k)} - \frac{p_{\sigma_i(k)i}}{2\mu_{\sigma_i(k)i}}\right) \leq \left(C_{x} - \frac{p_{xi}}{2\mu_{xi}}\right)\sum_{k = 1}^{j_i} p_{\sigma_i(k)i} \label{eq:compTimeInUnifLP}
\end{align}
which in turn implies $\textstyle\sum_{k = 1}^{j_i} p_{\sigma_i(k)i}/\mu_i \leq 2C_x - p_{xi}/\mu_{xi}$.

If all subjobs are released at time zero, then we can combine this with Lemma \ref{lem:CompTimesOnUniformMachines} and the fact that $p_{xit^*} \leq p_{xi} = \textstyle\sum_{t \in T_{xi}} p_{xit}$ to see the following (the transition from the first inequality the second inequality uses $C_x \geq p_{xit^*}/v_{1i}$ and $R_i = v_{1i}/\bar{v}_i$).
\begin{align}
\hat{C}_{xi} 
	&\leq 2C_x - \frac{p_{xi}}{\mu_{xi}} + \frac{p_{xit^*}}{\bar{v}_i} - \frac{p_{xit^*}}{\mu_i} \leq 
		C_x(2 + \left[R_i(1 - 2/m_i)\right]^+) \label{eq:generalCompTimeWithOUTReleaseUnif} 
\end{align}

When one or more subjobs are released after time zero, Lemma \ref{lem:CompTimesOnUniformMachines} implies that it is sufficient to bound  $\displaystyle\max_{1 \leq k \leq j_i}{\left\lbrace r_{\sigma_i(k)i} \right\rbrace}$ by some constant multiple of $C_x$. Since $\sigma_i$ is defined by increasing $L_{ji} \doteq C_j - p_{ji}/(2\mu_{ji})$, $L_{\sigma_i(a)i} \leq L_{\sigma_i(b)i}$ implies
\begin{align}
&r_{\sigma_i(a)i} + \frac{p_{\sigma_i(a)i}}{2\mu_{\sigma_i(a)i}} + \frac{p_{\sigma_i(b)i}}{2\mu_{\sigma_i(b)i}}  \leq C_{\sigma_i(a)} - \frac{p_{\sigma_i(a)i}}{2\mu_{\sigma_i(a)i}} + \frac{p_{\sigma_i(b)i}}{2\mu_{\sigma_i(b)i}}  \leq C_{\sigma_i(b)} ~\forall~ a \leq b
\end{align}
and so $\max_{1 \leq k \leq j_i}{\left\lbrace r_{\sigma_i(k)i} \right\rbrace} + p_{xi}/(2\mu_{xi}) \leq C_{x}$. As before, combine this with Lemma \ref{lem:CompTimesOnUniformMachines} and the fact that $p_{xit^*} \leq p_{xi} = \textstyle\sum_{t \in T_{xi}} p_{xit}$ to yield the following inequalities
\begin{align}
\hat{C}_{xi} 
	&\leq 3C_x - \frac{3p_{xi}}{2\mu_{xi}} + \frac{p_{xit^*}}{\bar{v}_i} - \frac{p_{xit^*}}{\mu_i} \leq C_x(3 + \left[R_i(1 - 5/(2m_i))\right]^+)  \label{eq:generalCompTimeWithReleaseUnif} 
\end{align}
-which complete our proof.
\end{proof}
\subsection{CC-LP for Identical Machines}\label{sec:identLP}
\begin{theorem}
If machines are of unit speed, then CC-LP yields an objective that is...
\begin{center}
\begin{tabular}{l | c c}
\hline
 & $r_{ji} \equiv 0$ & some $r_{ji} > 0$ \\ 
 \hline
single-task subjobs & $\leq$ 2 $OPT$ & $\leq $ 3 $OPT$ \\
multi-task subjobs & $\leq$ 3 $OPT$ & $\leq$ 4 $OPT$ \\
\hline
\end{tabular}
\end{center}
\label{thm:identLP}
\end{theorem}
\begin{proof}
Define $[\cdot]^+$, $x$, $C_x$, $\hat{C}_x$, $i$, $\sigma_i$, and $t^*$ as in Theorem \ref{thm:uniformLP1}. When $r_{ji} \equiv 0$, one need only give a more careful treatment of the first inequality in \eqref{eq:generalCompTimeWithOUTReleaseUnif} (using $\mu_{ji} = q_{ji}$).
\begin{align}
\hat{C}_{x,i} 
	&\leq 2C_x + p_{xit^*} - p_{xit^*}/m_i - p_{xi}/q_{xi} 
	\leq C_x(2 + \left[1 - 1/m_i -1/q_{xi} \right]^+) \label{eq:GeneralIdentCompTimeBound}
\end{align}
Similarly, when some $r_{ji} > 0$, the first inequality in \eqref{eq:generalCompTimeWithReleaseUnif} implies the following.
\begin{align}
\hat{C}_{x,i}
	&\leq 3C_x + p_{xit^*} - p_{xit^*}/m_i - 3p_{xi}/(2q_{xi})
	\leq C_x(3 + \left[1 - 1/m_i - 3/(2q_{xi})\right]^+) \label{eq:forCstTimeThmWithRelease}
\end{align}
\end{proof}
The key in the refined analysis of Theorem \ref{thm:identLP} lay in how $-p_{xi}/q_{xi}$ is used to annihilate $+p_{xit^*}$. While $q_{xi} = 1$ (i.e. single-task subjobs) is sufficient to accomplish this, it is not strictly \textit{necessary}. The theorem below shows that we can annihilate the $+p_{xit^*}$ term whenever all tasks of a given subjob are of the same length. Note that the tasks need not be \textit{unit}, as the lengths of tasks across different subjobs can differ.
\begin{theorem}
Suppose $v_{\ell i} \equiv 1$. If $p_{jit}$ is constant over $t \in T_{ji}$ for all $j \in N$ and $i \in M$, then algorithm CC-LP is a 2-approximation when $r_{ji} \equiv 0$, and a 3-approximation otherwise. \label{thm:identLP_2appxWithConstantTasks}
\end{theorem}
\begin{proof}
The definition of $p_{xi}$ gives $p_{xi}/q_{xi} = \textstyle\sum_{t \in T_{xi}} p_{xit} / q_{xi}$. Using the assumption that $p_{jit}$ is constant over $t \in T_{ji}$, we see that $p_{xi}/q_{xi} = (q_{xi} + |T_{xi}| - q_{xi})p_{xit^*} / q_{xi} $, where $|T_{xi} |\geq q_{xi}$. Apply this to Inequality \eqref{eq:GeneralIdentCompTimeBound} from the proof of Theorem \ref{thm:identLP}; some algebra yields 
\begin{align}
\hat{C}_{xi} 
	\leq& 2C_x  - p_{xit^*}/m_i - p_{xit^*}\left(|T_{xi}| - q_{xi}\right)/q_{xi} \leq 2C_x.
\end{align}
The case with some $r_{ji} > 0$ uses the same identity for $p_{xi}/q_{xi}$.
\end{proof}
Sachdeva and Saket \cite{nphard2} showed that it is NP-Hard to approximate $CC|m_i \equiv 1|\sum w_j C_j$ with a constant factor less than 2. 
Theorem \ref{thm:identLP_2appxWithConstantTasks} is significant because it shows that CC-LP can attain the same guarantee for \textit{arbitrary} $m_i$, provided $v_{\ell i} \equiv 1$ and $p_{jit}$ is constant over $t$.

\section{Combinatorial Algorithms}\label{sec:TSPT}




In this section, we introduce an extremely fast combinatorial algorithm with performance guarantees similar to CC-LP for ``unstructured'' inputs (i.e. those for which some $v_{\ell i} > 1$, or some $T_{ji}$ have $p_{jit}$ non-constant over $t$). 
We call this algorithm \textit{CC-TSPT}. 
CC-TSPT uses the MUSSQ algorithm for concurrent open shop (from \cite{mqssu}) as a subroutine. As SWAG (from \cite{HGY}) motivated development of CC-TSPT, we first address SWAG's worst-case performance.

\subsection{A Degenerate Case for SWAG}\label{subsec:swagDegenerate}
\begin{wrapfigure}{r}{0.5\textwidth}
	\vspace{-0.45cm}
  \centering
    \begin{minipage}{0.50\textwidth}
	\begin{algorithm}[H]
	\begin{algorithmic}[1]
	\Procedure{SWAG}{$N,M,p_{ji}$}
	\State $J\gets \emptyset$
	\State $q_i\gets 0,\forall i\in M$
	\While{$|J|\not=|N|$}
	\State mkspn$_j\gets$ max$_{i\in M}\left(\frac{q_i+p_{ji}}{m_i}\right)$
	\item[] \qquad \qquad $\forall j\in N\setminus J$
	\State nextJob $\gets$ argmin$_{j \in N \setminus J}\ $mkspn$_j$
	\State $J.$append$($nextJob$)$
	\State $q_i\gets q_i+ p_{ji}$
	\EndWhile
	\State \textbf{return} $J$
	\EndProcedure
	\end{algorithmic}
	\end{algorithm}
  \end{minipage}
  \vspace{-0.45cm}
\end{wrapfigure}
As a prerequisite for addressing worst-case performance of an existing algorithm, we provide psuedocode and an accompanying verbal description for SWAG.

SWAG computes queue positions for every subjob of every job, supposing that each job was scheduled next. 
A job's potential makespan (``mkspn'') is the largest of the potential finish times of all of its subjobs (considering current queue lengths $q_i$ and each subjob's processing time $p_{ji}$). 
Once potential makespans have been determined, the job with smallest potential makespan is selected for scheduling. 
At this point, all queues are updated. 
Because queues are updated, potential makespans will need to be re-calculated at the next iteration. 
Iterations continue until the very last job is scheduled. Note that SWAG runs in $O(n^2m)$ time.

\begin{theorem}
For an instance $I$ of $PD || \sum C_j$, let $SWAG(I)$ denote the objective function value of SWAG applied to $I$, and let $OPT(I)$ denote the objective function value of an optimal solution to $I$. 
Then for all $L \geq 1$, there exists an $I \in \Omega_{PD || \sum C_j}$ such that $SWAG(I) / OPT(I) > L$.
\label{thm:swagBad}
\end{theorem}
\begin{proof}
Let $L \in \mathbb{N}^+$ be a fixed but arbitrary constant. 
Construct a problem instance $I_L^m$ as follows: 

$N = N_1 \cup N_2$ where $N_1$ is a set of $m$ jobs, and $N_2$ is a set of $L$ jobs. 
Job $j \in N_1$ has processing time $p$ on cluster $j$ and zero all other clusters. 
Job $j \in N_2$ has processing time $p(1-\epsilon)$ on all $m$ clusters. 
$\epsilon$ is chosen so that $\epsilon < 1/L$
(see Figure \ref{fig:swag1}).

\begin{figure}[ht]
\includegraphics[width=\linewidth]{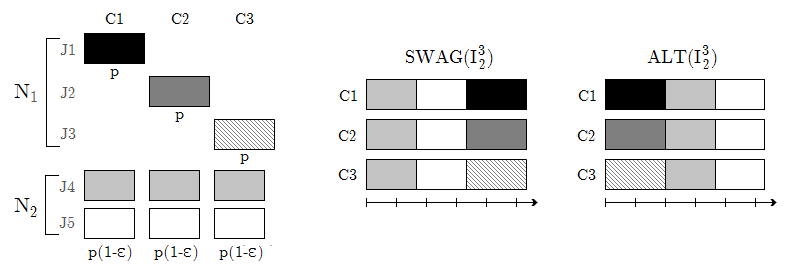}
\centering
\caption{At left, an input for SWAG example with $m=3$ and $L=2$. At right, SWAG's resulting schedule, and an alternative schedule.}
\label{fig:swag1}
\end{figure}

It is easy to verify that SWAG will generate a schedule where all jobs in $N_2$ precede all jobs in $N_1$ (due to the savings of $p \epsilon$ for jobs in $N_2$). 
We propose an \textit{alternative} solution in which all jobs in $N_1$ preceed all jobs in $N_2$. 
Denote the objective value for this alternative solution $ALT(I_L^m)$, noting $ALT(I_L^m) \geq OPT(I_L^m)$.

By symmetry, and the fact that all clusters have a single machine, we can see that $SWAG(I_L^m)$ and $ALT(I_L^m)$ are given by the following
\begin{align}
SWAG(I_L^m) &= p(1-\epsilon)L(L+1)/2 + p(1-\epsilon)L m + p m \\
ALT(I_L^m) &= p(1-\epsilon)L(L + 1)/2 + pL + p m
\end{align}
Since $L$ is fixed, we can take the limit with respect to $m$.
\begin{align}
\lim_{m \rightarrow \infty}{\frac{SWAG(I_L^m)}{ALT(I_L^m)}} 
	&= \lim_{m \rightarrow \infty}{\frac{p(1-\epsilon)L m + p m}{p m}} = L(1-\epsilon) + 1 > L
\end{align}
The above implies the existence of a sufficiently large number of clusters $\overline{m}$, such that $m \geq \overline{m}$ implies $ SWAG(I_L^{m})/OPT(I_L^{m}) > L $. This completes our proof.
\end{proof}
Theorem \ref{thm:swagBad} demonstrates that that although SWAG performed well in simulations, it may not be reliable. 
The rest of this section introduces an algorithm not only with superior runtime to SWAG (generating a permutation of jobs in $O(n^2 + nm)$ time, rather than $O(n^2m)$ time), but also a constant-factor performance guarantee.

\subsection{CC-TSPT : A Fast 2 + R Approximation}\label{subsec:fastreduction}
Our combinatorial algorithm for concurrent cluster scheduling exploits an elegant transformation to concurrent open shop. 
Once we consider this simpler problem, it can be handled with MUSSQ \cite{mqssu} and List-LPT. 
Our contributions are twofold: (1) we prove that this intuitive technique yields an approximation algorithm for a decidedly more general problem, and (2) we show that a \textit{non-intuitive} modification can be made that maintains theoretical bounds while improving empirical performance. 
We begin by defining our transformation.

\begin{definition}[The Total Scaled Processing Time (TSPT) Transformation]
Let $\Omega_{CC}$ be the set of all instances of $CC || \sum w_j C_j$, 
and let $\Omega_{PD}$ be the set of all instances of 
$PD || \sum w_j C_j$. Note that $\Omega_{PD} \subset \Omega_{CC}$.
Then the Total Scaled Processing Time Transformation is a mapping
\begin{align*}
TSPT: ~\Omega_{CC} \to \Omega_{PD} \quad \text{ with } \quad (T, v, w) &\mapsto (X, w) ~:~ x_{ji} = \textstyle\sum_{t \in T_{ji}} p_{jit} / \mu_i
\end{align*}
i.e., $x_{ji}$ is the total processing time required by subjob $(j,i)$, scaled by the sum of machine speeds at cluster $i$. 
Throughout this section, we will use $I = (T, v, w)$ to denote an arbitrary instance of $CC || \sum w_j C_j$, and $I' = (X, w)$ as the image of $I$ under TSPT. 
Figure \ref{fig:tspt} shows the result of TSPT applied to our baseline example. 
\end{definition}

\begin{figure}[ht]
\centering
\includegraphics[width=0.8\linewidth]{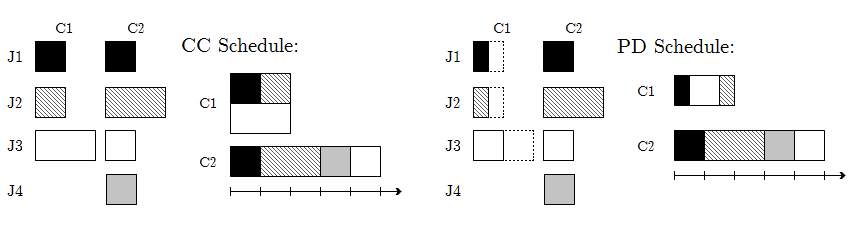}
\caption{An instance $I$ of $CC||\sum w_j C_j$, and its image $I' = TSPT(I)$. The schedules were constructed with List-LPT using the same permutation for $I$ and $I'$. }
\label{fig:tspt}
\end{figure}

We take the time to emphasize the simplicity of our reduction. Indeed, the TSPT transformation is perhaps the first thing one would think of given knowledge of the concurrent open shop problem. What is surprising is how one can attain constant-factor performance guarantees even after such a simple transformation.

\vspace{1em}
\noindent \textbf{Algorithm CC-TSPT} : Execute MUSSQ on $I'= TSPT(I)$ to generate a permutation of jobs $\sigma$. List schedule instance $I$ by 
$\sigma$ on each cluster according to List-LPT.
\vspace{1em}

Towards proving the approximation ratio for CC-TSPT, we will establish a critical inequality in Lemma \ref{lem:TSPTboundInTermsOfPDandOPT}. 
The intuition behind Lemma \ref{lem:TSPTboundInTermsOfPDandOPT} requires thinking of every job $j$ in $I$ as having a corresponding representation in $j'$ in $I'$. 
Job $j$ in $I$ will be scheduled in the $CC$ environment, while job $j'$ in $I'$ will be scheduled in the $PD$ environment. 
We consider what results when the same permutation $\sigma$ is used for scheduling in both environments. 

Now the definitions for the lemma: let $C^{CC}_{\sigma(j)}$ be the completion time of job $\sigma(j)$ resulting from List-LPT on an arbitrary permutation $\sigma$. 
Define $C^{CC\star}_{\sigma(j)}$ as the completion time of job $\sigma(j)$ in the $CC$ environment in the optimal solution. 
Lastly, define $C^{PD,I'}_{\sigma(j')}$ as the completion time of job $\sigma(j')$ in $I'$ when scheduling by List-LPT($\sigma$) in the $PD$ environment.

\begin{lemma}
For $I' = TSPT(I)$, let $j'$ be the job in $I'$ corresponding to job $j$ in $I$. 
For an arbitrary permutation of jobs $\sigma$, we have $C^{CC}_{\sigma(j)} \leq C^{PD,I'}_{\sigma(j')} + R\cdot C^{CC\star}_{\sigma(j)}$. \label{lem:TSPTboundInTermsOfPDandOPT}
\end{lemma}
\begin{proof}
After list scheduling has been carried out in the $CC$ environment, we may determine $C^{CC}_{\sigma(j)i}$ - the completion time of subjob $(\sigma(j),i)$. 
We can bound $C^{CC}_{\sigma(j)i}$ using Lemma \ref{lem:CompTimesOnUniformMachines} (which implies \eqref{eqInLem:TSPTbound1}), and the serial-processing nature of the $PD$ environment (which implies \eqref{eqInLem:TSPTbound2}).
\begin{align}
& C^{CC}_{\sigma(j)i} \leq p_{\sigma(j)i1}\left(1/\bar{v} - 1/\mu_i\right) + \textstyle\sum_{\ell = 1}^{j} p_{\sigma(\ell)i}/\mu_i  \label{eqInLem:TSPTbound1} \\
& \textstyle\sum_{\ell = 1}^j p_{\sigma(\ell)i}/\mu_i \leq C^{PD,I'}_{\sigma(j')} \quad \forall ~i \in M \label{eqInLem:TSPTbound2}
\end{align}
If we relax the bound given in Inequality \eqref{eqInLem:TSPTbound1} and combine it with Inequality \eqref{eqInLem:TSPTbound2}, we see that $C^{CC}_{\sigma(j)i} \leq C^{PD,I'}_{\sigma(j')}  +  p_{\sigma(j)i1}/\bar{v}$. 
The last step is to replace the final term with something more meaningful. Using $p_{\sigma(j)1}/\bar{v} \leq R \cdot C^{CC\star}_{\sigma(j)}$ (which is immediate from the definition of $R$) the desired result follows.
\end{proof}
While Lemma \ref{lem:TSPTboundInTermsOfPDandOPT} is true for arbitrary $\sigma$, now we consider $\sigma = MUSSQ(X, w)$. 
The proof of MUSSQ's correctness established the first inequality in the chain of inequalities below. 
The second inequality can be seen by substituting $p_{ji} / \mu_{i}$ for $x_{ji}$  in LP0($I'$) (this shows that the constraints in LP0($I'$) are weaker than those in LP1($I$)). 
The third inequality follows from the Validity Lemma.
\begin{equation}
\textstyle\sum_{j \in N} w_{\sigma(j)} C^{PD,I'}_{\sigma(j)} 
	\leq 2 \textstyle\sum_{j \in N} w_j C^{\text{LP0}(I')}_j
	\leq 2 \textstyle\sum_{j \in N} w_j C^{\text{LP1}(I)}_j
	\leq 2 OPT(I) \label{eq:tsptCore}
\end{equation}
Combining Inequality \eqref{eq:tsptCore} with Lemma \ref{lem:TSPTboundInTermsOfPDandOPT} allows us to bound the objective in a way that does not make reference to $I'$.
\begin{equation}
\textstyle\sum_{j \in N} w_{\sigma(j)}C^{CC}_{\sigma(j)} 
	\leq \textstyle\sum_{j \in N} w_{\sigma(j)}\left[C^{PD,I'}_{\sigma(j)} + R\cdot C^{CC\star}_{\sigma(j)}\right] \leq ~ 2 \cdot OPT(I) + R \cdot OPT(I) \label{eq:dontReferenceIPrime}
\end{equation}
Inequality \eqref{eq:dontReferenceIPrime} completes our proof of the following theorem.
\begin{theorem}
Algorithm CC-TSPT is a $2 + R$ approximation for $CC || \sum w_j C_j$.
\label{thm:algCCTspt}
\end{theorem}

\subsection{CC-TSPT with Unit Tasks and Identical Machines}\label{subsec:tsptOnUnitTasks}
Consider concurrent cluster scheduling with $v_{\ell i} = p_{jit} = 1$ (i.e., all processing times are unit, although the size of the collections $T_{ji}$ are unrestricted). In keeping with the work of Zhang, Wu, and Li \cite{zwl} (who studied this problem in the single-cluster case), we call instances with these parameters ``fully parallelizable,'' and write $\beta = fps$ for Graham's $\alpha|\beta|\gamma$ taxonomy.

Zhang et al. showed that scheduling jobs greedily by  ``Largest Ratio First'' (decreasing $w_j / p_{j}$) results in a 2-approximation, where 2 is a tight bound. 
This comes as something of a surprise since the Largest Ratio First policy is \textit{optimal} for $1||\sum w_j C_j~$- which their problem very closely resembles. 
We now formalize the extent to which $P|fps|\sum w_j C_j$ resembles $1||\sum w_j C_j~$: define the \textit{time resolution} of an instance $I$ of $CC |fps| \sum w_jC_j$ as $ \rho_I = \min_{j \in N, i \in M}{\big\lceil{p_{ji}/m_i}\big\rceil}$. 
Indeed, one can show that as the time resolution increases, the performance guarantee for LRF on $P | fps | \sum w_j C_j$ approaches that of LRF on $1||\sum w_j C_j$. 
We prove the analogous result for our problem. 
\begin{theorem}
CC-TSPT for $CC |fps| \sum w_jC_j$ is a $(2 + 1/\rho_I)-$approximation.
\label{thm:tspt_unit_tasks}
\end{theorem}
\begin{proof}
Applying techniques from the proof of Lemma \ref{lem:TSPTboundInTermsOfPDandOPT} under the hypothesis of this theorem, we have $C^{CC}_{\sigma(j), i} \leq C^{PD,I'}_{\sigma(j)} + 1$. 
Next, use the fact that for all $j \in N$, $C^{CC,OPT}_{\sigma(j)} \geq \rho_I$ by the definition of $\rho_I$. These facts together imply
$C^{CC}_{\sigma(j), i} \leq C^{PD,I'}_{\sigma(j)} + C^{CC,OPT} / \rho_I$. Thus
\begin{align}
\textstyle\sum_{j \in N} w_j C^{CC}_{\sigma(j)} 
	&\leq \textstyle\sum_{j \in N} w_j \left[C^{PD,I'}_{\sigma(j)} + C^{CC,OPT} / \rho_I\right] \leq 2 \cdot OPT + OPT / \rho_I.
\end{align}
\end{proof}

\subsection{CC-ATSPT : Augmenting the LP Relaxation}
The proof of Theorem \ref{thm:algCCTspt} appeals to a trivial lower bound on $C^{CC\star}_{\sigma(j)}$, namely $p_{\sigma(j)1}/\bar{v} \leq R \cdot C^{CC\star}_{\sigma(j)}$. We attain constant-factor performance guarantees in spite of this, but it is natural to wonder how the \textit{need} for such a bound might come hand-in-hand with empirical weaknesses. Indeed, TSPT can make subjobs consisting of many small tasks look the same as subjobs consisting of a single very long task.
Additionally, a cluster hosting a subjob with a single extremely long task might be identified as a bottleneck by MUSSQ, even if that cluster has more machines than it does tasks to process.

We would like to mitigate these issues by introducing the simple lower bounds on $C_j$ as seen in constraints $(1B)$ and $(1C)$. This is complicated by the fact that MUSSQ's proof of correctness only allows constraints of the form in $(1A)$. For $I \in \Omega_{PD}$ this is without loss of generality, since $|S| = 1$ in LP0 implies $C_j \geq p_{ji}$, but since we apply LP0 to $I' = TSPT(I)$, $C_j \geq x_{ji}$ is equivalent to $C_j \geq p_{ji}/\mu_i$ (a much weaker bound than we desire). 

Nevertheless, we can bypass this issue by introducing additional clusters and appropriately defined subjobs. We formalize this with the ``Augmented Total Scaled Processing Time'' (ATSPT) transformation. 
Conceptually, ATSPT creates $n$ ``imaginary clusters'', where each imaginary cluster has nonzero processing time for exactly one job.
\begin{definition}[The Augmented TSPT Transformation]
Let $\Omega_{CC}$ and $\Omega_{PD}$ be as in the definition for TSPT. Then the Augmented TSPT Transformation is likewise a mapping
\begin{align*}
ATSPT: ~\Omega_{CC} \to \Omega_{PD} \quad \text{ with } \quad (T, v, w) &\mapsto (X, w) ~:~ X = \big[\begin{array}{c|c} X_{TSPT(I)} & D \end{array}\big].
\end{align*}
Where $D \in \mathbb{R}^{n \times n}$ is a diagonal matrix with $d_{jj}$ as any valid lower bound on the completion time of job $j$ (such as the right hand sides of constraints ($1B$) and ($1C$) of LP1).
\end{definition}
Given that $d_{jj}$ is a valid lower bound on the completion time of job $j$, it is easy to verify that for $I' = ATSPT(I)$, LP1($I'$) is a valid relaxation of $I$. 
Because MUSSQ returns a permutation of jobs for use in list scheduling by List-LPT, these ``imaginary clusters'' needn't be accounted for beyond the computations in MUSSQ.

\section{A Reduction for Minimizing Total Weighted Lateness on Identical Parallel Machines }\label{sec:relationshipsBetweenProbs}
The problem of minimizing total weighted lateness on a bank of identical parallel machines is typically denoted $P || \sum w_jL_j$, where the lateness of a job with deadline $d_j$ is $L_j \doteq \max{\{C_j - d_j, 0\}}$. The reduction we offer below shows that $P || \sum w_j L_j$ can be stated in terms of $CC || \sum w_jC_j$ \textit{at optimality}. Thus while a $\Delta$ approximation to $CC || \sum w_jC_j$ does not imply a $\Delta$ approximation to $P || \sum w_j L_j$, the reduction below nevertheless provides new insights on the structure of $P || \sum w_j L_j$.

\begin{definition}[Total Weighted Lateness Reduction]
Let $I = (p, d, w, m)$ denote an instance of $P || \sum w_j L_j$. 
$p$ is the set of processing times, $d$ is the set of deadlines, $w$ is the set of weights, 
and $m$ is the number of identical parallel machines. 
Given these inputs, we transform $I \in \Omega_{P || \sum w_j L_j}$ 
to $I' \in \Omega_{CC}$ in the following way.

Create a total of $n + 1$ clusters. Cluster 0 has $m$ machines. Job $j$ has processing time $p_j$ on this cluster, and $|T_{j0}| = 1$. Clusters 1 through $n$ each consist of a single machine. Job $j$ has processing time $d_j$ on cluster $j$, and zero on all clusters other than cluster 0 and cluster $j$. Denote this problem $I'$.
\end{definition}
We refer the reader to Figure \ref{fig:probstm3and4} for an example output of this reduction.
\begin{theorem}
Let $I$ be an instance of $P || \textstyle\sum w_j L_j$. Let $I'$ be an instance of $CC|| \sum w_j C_j$ resulting from the transformation described above. Any list schedule $\sigma$ that is optimal for $I'$ is also optimal for $I$.
\end{theorem}
\begin{proof}
If we restrict the solution space of $I'$ to single permutations (which we may do without loss of generality), then any schedule $\sigma$ for $I$ or $I'$ produces the same value of $\sum_{j \in N}  w_j(C_j - d_j)^+$ for $I$ and $I'$.
The additional clusters we added for $I'$ ensure that $C_j \geq d_j$. Given this, the objective for $I$ can be written as $\sum_{j \in N} w_j d_j + w_j(C_j - d_j)^+$. Because $w_j d_j$ is a constant, any permutation to solve $I'$ optimally also solves $\sum_{j \in N} w_j (C_j - d_j)^+$ optimally. Since $\sum_{j \in N} w_j (C_j - d_j)^+ = \sum_{j \in N} w_j L_j$, we have the desired result.
\end{proof}
\section{Closing Remarks}\label{sec:discAndConc}
We now take a moment to address a subtle issue in the concurrent cluster problem: what price do we pay for using the same permutation on all clusters (i.e. single-$\sigma$ schedules)? For concurrent open shop, it has been shown (\cite{Sris1993, mqssu}) that single-$\sigma$ schedules may be assumed without loss of optimality. As is shown in Figure \ref{fig:singleVsMultiPerm}, this does \textit{not} hold for concurrent cluster scheduling in the general case. In fact, that is precisely why the strong performance guarantees for algorithm CC-LP rely on clusters having possibly unique permutations.

\begin{figure}[ht]
\centering
\includegraphics[width=0.7\textwidth]{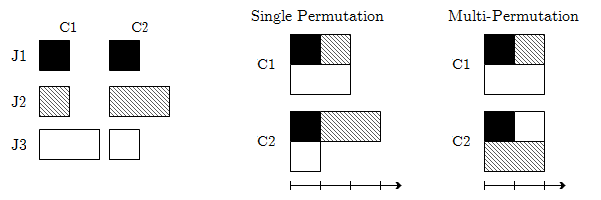}
\caption{An instance of $CC||\sum C_j$ (i.e. $w_j \equiv 1$) for which there does not exist a single-$\sigma$ schedule which attains the optimal objective value. In the single-$\sigma$ case, one of the jobs necessarily becomes delayed by one time unit compared to the multi-$\sigma$ case. As a result, we see a 20\% optimality gap even when $v_{\ell i } \equiv 1$.}\label{fig:singleVsMultiPerm}
\centering
\end{figure}

Our more novel contributions came in our analysis for CC-TSPT and CC-ATSPT. First, we could not rely on the processing time of the last task for a job to be bounded above by the job's completion time variable $C_j$ in LP0($I'$), and so we appealed to a lower bound on $C_j$ that was not stated in the LP itself. The need to incorporate this second bound is critical in realizing the strength of algorithm CC-TSPT, and uncommon in LP rounding schemes. Second, CC-ATSPT is novel in that it introduces constraints that would be redundant for LP0($I$) when $I \in \Omega_{PD}$, but become relevant when viewing $LP0(I')$ as a relaxation for $I \in \Omega_{CC}$. This approach has potential for more broad applications since it represented effective use of a limited constraint set supported by a known primal-dual algorithm.

We now take a moment to state some open problems in this area. One topic of ongoing research is developing a factor 2 purely combinatorial algorithm for the special case of concurrent cluster scheduling considered in Theorem \ref{thm:identLP_2appxWithConstantTasks}. In addition, it would be of broad interest to determine the worst-case loss to optimality incurred by assuming single-permutation schedules for $CC|v\equiv 1|\sum w_j C_j$. The simple example above shows that an optimal single-$\sigma$ schedule can have objective 1.2 times the globally optimal objective. Meanwhile, Theorem \ref{thm:algCCTspt} shows that there always exists a single-$\sigma$ schedule with objective no more than 3 times the globally optimal objective. Thus, we know that the worst-case performance ratio is in the interval $[1.2,3]$, but we do not know its precise value. As a matter outside of scheduling theory, it would be valuable to survey primal-dual algorithms with roots in LP relaxations to determine which have constraint sets that are amenable to implicit modification, as in the fashion of CC-ATSPT.

\subparagraph*{Acknowledgments.}

Special thanks to Andreas Schulz for sharing some of his recent work with us \cite{Schulz2012}. His thorough analysis of a linear program for $P||\sum w_j C_j$ drives the LP-based results in this paper. Thanks also to Chien-Chung Hung and Leana Golubchik for sharing \cite{HGY} while it was under review, and to Ioana Bercea and Manish Purohit for their insights on SWAG's performance. Lastly, our sincere thanks to William Gasarch for organizing the REU which led to this work, and to the 2015 CAAR-REU cohort for making the experience an unforgettable one; in the words of Rick Sanchez \textit{wubalubadubdub!}

\bibliography{p234-murray}

\end{document}